\documentclass[12pt,journal,onecolumn]{IEEEtran}
\usepackage{amssymb}
\usepackage{graphicx}
\usepackage{amsmath}
\usepackage{epsf}
\usepackage{mathrsfs}
\usepackage{cite}
\usepackage{dsfont}

 \usepackage[usenames,dvipsnames]{pstricks}
 \usepackage{epsfig}
 \usepackage{pst-grad} 
 \usepackage{pst-plot} 
 \usepackage{pst-node}
 \usepackage{pst-3d}
\newtheorem{theorem}{Theorem}

\newtheorem{definition}{Definition}

\newtheorem{lemma}{Lemma}

\newtheorem{remark}{Remark}

\begin{document}

\title{Real Interference Alignment with Real Numbers}

\author{ Abolfazl~Seyed~Motahari, ~Shahab Oveis Gharan,~ and~Amir~Keyvan~Khandani
\\\small Electrical and Computer Engineering Department, University of Waterloo
\\\small Waterloo, ON, Canada N2L3G1
\\ \{abolfazl,shahab,khandani\}@cst.uwaterloo.ca
}



\maketitle

\let\thefootnote\relax\footnotetext{Financial support provided by Nortel and the
corresponding matching funds by the Natural Sciences and Engineering
Research Council of Canada (NSERC), and Ontario Centers of
Excellence (OCE) are gratefully acknowledged.}

\begin{abstract}
A novel coding scheme applicable in networks with single antenna nodes is proposed. This scheme converts a single
antenna system to an equivalent Multiple Input Multiple Output (MIMO) system with fractional dimensions. Interference
can be aligned along these dimensions and higher Multiplexing gains can be achieved. Tools from the field of
Diophantine approximation in number theory are used to show that the proposed coding scheme in fact mimics the
traditional schemes used in MIMO systems where each data stream is sent along a direction and alignment happens when
several streams arrive at the same direction. Two types of constellation are proposed for the encoding part, namely
the single layer constellation and the multi-layer constellation.

Using the single layer constellation, the coding scheme is applied to the two-user $X$ channel and the three-user
Gaussian Interference Channel (GIC). In case of the two-user $X$ channel, it is proved that the total
Degrees-of-Freedom (DOF), i.e. $\frac{4}{3}$, of the channel is achievable almost surely. This is the first example
in which it is shown that a time invariant single antenna system does not fall short of achieving its total DOF. For
the three-user GIC, it is shown that the DOF of $\frac{4}{3}$ is achievable almost surely.

Using the multi-layer constellation, the coding scheme is applied to the symmetric three-user GIC. Achievable DOFs
are derived for all channel gains. As a function of the channel gain, it is observed that the DOF is everywhere
discontinuous. In particular, it is proved that for the irrational channel gains the achievable DOF meets the upper
bound $\frac{3}{2}$. For the rational gains, the achievable DOF has a gap to the available upper bounds. By allowing
carry over from multiple layers, however, it is shown that higher DOFs can be achieved.
\end{abstract}

\begin{keywords}
Interference channels, interference alignment, number theory, fractional dimensions, Diophantine approximation.
\end{keywords}

\newpage

\section{Introduction}
\PARstart{I}{nterference management} plays a crucial role in future wireless systems as
the number of users sharing the same medium is rapidly growing. In fact, an
increase in the number of users results in an increase in the amount of
interference in the system. This interference may cause a severe degradation
in the system's performance.

The study of interaction between two users sharing the same channel goes
back to Shannon's work on the two-way channel in \cite{TWOWAY:SHANNON}. His
work was followed by several researchers and the two-user interference
channel emerged as the fundamental problem regarding interaction between users
causing interference in the networks.

The problem of characterizing  the capacity region of the two-user Gaussian Interference Channel (GIC) is still open.
In \cite{Etkin}, a major step is taken and the region is characterized within one bit. Followed by this work, the sum
capacity of the two-user GIC is derived in low Signal to Noise Ratios (SNR), see \cite{abolfazl,Shang,Annapureddy}.
Interestingly, it is proved that treating interference as noise is optimal within the given range of SNRs.

In high SNR regimes, if the interference is ignored and
treated as noise then the throughput of the system decreases dramatically. In particular, as the number of
active interfering users increases the interference becomes more and more severe and the throughput drops rapidly.
However, this contradicts the actual behavior of the system as recent results show that the throughput is  constant
regardless of the number of active users in the system, c.f. \cite{cadambe2008iaa}.

Interference alignment is a solution for making the interference less severe at receivers. In
\cite{maddahali2008com}, Maddah-Ali, Motahari, and Khandani pioneered the concept of
interference alignment and showed its capability in achieving the
full Degrees-Of-Freedom (DOF) of a class of two-user $X$ channels. Being simple and powerful at the same time,
interference alignment provided the spur for further research. Not only usable for lowering the harmful effect of
interference, but also it can be applied to provide security in networks as proposed in \cite{Ozan-gamal-lai-poor}.

Interference alignment in $n$-dimensional Euclidean spaces for $n\geq 2$ is studied by several researchers, c.f.
\cite{maddahali2008com,jafar2008dfr,cadambe2008iaa,cadambe2008dfw}. This method can be applied, for example, by
choosing a specific subspace for interference, and forcing all interfering transmitters
to send data such that it is received at the pre-assigned subspace in the receiver. Using this method, Cadambe and
Jafar showed that a $K$-user Gaussian interference channel with varying channel gains can achieve its total DOF which
is $\frac{K}{2}$.

Application of interference alignment is not confined to two or more dimensional spaces. In fact, it can be applied
in one-dimensional spaces as well, c.f. \cite{Bresler-Parekh-tse,Sridharan,sridharan3llc}. In \cite{Sridharan}, after
aligning interference using lattice codes the aggregated signal is decoded and its effect is subtracted from the
received signal. In fact, \cite{Sridharan} shows that the very-strong interference region of the $K$-user GIC is
strictly larger than the corresponding region when alignment is not applied. In this method, to make the interference
less severe, transmitters use lattice codes to reduce the code-rate of the interference which guarantees decodability
of the interference at the receiver. In \cite{sridharan3llc}, Sridharan et al. showed that the DOF of a class of
3-user GICs with fixed channel gains can be greater than one. This result obtained using layered lattice codes along
with successive decoding at the receiver.

The first examples of interference alignment in one-dimensional spaces are reported in \cite{Etkin-Ordentlich} and
\cite{abolfazl-shahab-amir} where the results from the filed of Diophantine approximation in number theory are used
to show that interference can be aligned using properties of rational and irrational numbers and their relations.
They showed that the total DOF of some classes of time-invariant single antenna interference channels can be
achieved. In particular, Etkin and Ordentlich in \cite{Etkin-Ordentlich} proposed an upper bound on the total DOF
which respects the properties of channel gains with respect to being rational or irrational.  Using this upper bound,
surprisingly, they proved that the DOF is everywhere discontinuous.

Built on \cite{abolfazl-shahab-amir} and \cite{Etkin-Ordentlich}, this paper broadens the applications of
interference alignment. In fact, we will show that it is possible to perform alignment in single dimensional systems
such as time-invariant networks equipped with single antennas at all nodes. In Section II, we summarize the main
contributions of this paper.

In Section III, we propose a novel coding scheme in which data streams are encoded using constellation points from
integers and transmitted in the directions of irrational numbers. Two types of constellation designs are considered,
namely the  single layer and the multi-layer constellations. It is shown that the coding provides sufficient tools to
accomplish interference alignment in one-dimensional spaces.

Throughout Section V, the single layer constellation is incorporated in the coding scheme. First, the performance of
a decoder is analyzed using the Khintchine-Groshev theorem in number theory. It is shown that under some regularity
conditions data streams can carry data with fractional multiplexing gains. The two-user $X$ channel is considered as
the first example in which the single layer constellation is incorporated in the coding schem. It is proved that for
this channel the total DOF of $\frac{4}{3}$ is attainable almost surely. For the $K$-user GIC, achievable DOFs are
characterized for some class of channels. Finally, it is proved that the DOF of $\frac{4}{3}$ is achievable for the
three-user GIC almost surely.

Throughout Section V, the multi-layer constellation is incorporated in the coding scheme. The channel under
investigation is the symmetric three-user GIC. An achievable DOF is derived for all channel gains. Viewed as a
function of the channel gain, this achievable DOF is everywhere discontinuous. It is shown that the total DOF of
$\frac{3}{2}$ is achievable for all irrational gains. For rational gains, the achievable rate has a gap to the
available upper bounds. In Section VII, we conclude the paper.

\textbf{Notation}: $\mathbb{R}$, $\mathbb{Q}$, $\mathbb{N}$ represent the set
of real, rational and nonnegative integers, respectively. For a real
number $x$, $\lfloor x\rfloor$ is the greatest integer less than
$x$ and $\lceil  x\rceil$ is the least integer greater than
$x$. For a random variable $X$, $E[X]$ denotes the expectation
value. $(m,n)$ represents the greatest common divisor of  two
integers $m$ and $n$. For two integers $m$ and $n$, $m|n$ means that $n$ is divisible by $m$. Similarly, $m\nmid n$
means that $n$ is not divisible by $m$. $[m\ n]$ denotes the set of integers between $m$ and $n$.

\section{Main Contributions}
In this paper, we are primarily interested in characterizing the total DOF of the two-user $X$ channel and the
$K$-user GIC. Let $\mathcal{C}$ denote the capacity region of the $K$-user GIC (a similar argument can be used for
the $X$ channel). The DOF region denoted by $\mathcal{R}$ associated with the channel is in fact the shape of
$\mathcal{C}$ in high SNR regimes scaled by $\log \text{SNR}$. All extreme points of $\mathcal{R}$ can be identified
by solving the following optimization problem:
\begin{equation}
r_{\boldsymbol{\lambda}}=\lim_{\text{SNR}\rightarrow\infty}\max_{\mathbf{R}\in
\mathcal{C}}\frac{\boldsymbol{\lambda}^t \mathbf{R}}{\log \text{SNR}}.
\end{equation}
The total DOF refers to the case where $\boldsymbol{\lambda}=\{1,1,\ldots,1\}$, i.e., the sum-rate is concerned.
Throughout this paper, $r_{\text{sum}}$ denotes the total DOF of the system. In what follows we summarize main
contributions of this paper regarding the total DOF of the $X$ channel and the $K$-user GIC.

\subsection{Bringing Another Dimension to Life: Rational Dimension}
Proposed in \cite{maddahali2008com}, the first example of interference alignment is done in Euclidean spaces.
Briefly, the $n$-dimensional Euclidean space ($n\geq2$) available at a receiver is partitioned into two subspaces. A
subspace is dedicated to interference and all interfering users are forced to respect this constraint. The major
technique is to reduce the dimension of this subspace so that the available dimension in the signal subspace allows
higher data rate for the intended user. Alignment using structural codes is also considered by several researchers
\cite{Bresler-Parekh-tse,sridharan3llc}. Structural interference alignment is used to make the interference caused by
users less severe by reducing the number of possible codewords at receivers. Even though useable in one-dimensional
spaces, this technique does not allow transmission of different data streams as there is only one dimension available
for transmission.

In this paper, we show that there exist available dimensions (called rational dimensions) in one-dimensional spaces
which open new ways of transmitting several data streams from a transmitter and interference alignment at receivers.
A coding scheme that provides sufficient tools to incorporate the rational dimensions in transmission is proposed.
This coding scheme relies on the fact that irrational numbers can play the role of directions in Euclidean spaces and
data can be sent by using rational numbers. This fact is proved by using the results of Hurwitz, Khintchine, and
Groshev obtained in the field of Diophantine Approximation. In the encoding part, two types of constellation are used
to modulate data streams. Type I or single layer constellation refers to the case where all integer points in an
interval are chosen as constellation points. Despite its simplicity, it is shown that the single layer constellation
is capable of achieving the total DOF of several channels. Type II or multi-layer constellation refers to the case
that a subset of integer points in an interval is chosen as constellation points. Being able of achieving the total
DOF of some channels, this constellation is more useful when all channel gains are rational.

\subsection{Breaking the Ice: Alignment in One dimension}
Obtained results regarding the total DOF of networks are based on interference alignment in $n$-dimensional Euclidean
spaces where $n\geq 2$, c.f.
\cite{maddahali2008com,jafar2008dfr,cadambe2008iaa,Cadambe-jarfar-wang,Gomadam-cadambe-jafar,Huang-jafar}. For
example in \cite{cadambe2008iaa}, the total DOF of the $K$-user Gaussian interference channel is derived when each
transmitter and receiver is equipped with a single antenna. In order to be able to align interference, however, it is
assumed that the channel is varying. This in fact means that nodes are equipped with multiple antennas and channel
coefficients are diagonal matrices.

Recently, \cite{Etkin-Ordentlich} and \cite{abolfazl-shahab-amir} independently reported  that the total DOF of some
classes of fully connected GICs can be achieved. Although being time invariant, these classes have measure zero with
respect to Lebesque measure. In this paper, we prove that the total DOF of time invariant two-user $X$ channel which
is $\frac{4}{3}$ can be attained almost surely. In other words, the set of channels that this DOF can not be achieved
has measure zero. This is done by incorporating rational dimensions in transmission. In fact, two independent data
streams from each transmitter are send while at each receiver two interfering streams are aligned. This achieves the
multiplexing gain of $\frac{1}{3}$ per data streams and the total of $\frac{4}{3}$ for the system. We also prove that
the same DOF can be achieved for the three-user GIC. However, for this case there is a gap between the available
upper bound, i.e. $\frac{3}{2}$, and the achievable DOF.

\subsection{$K$-user GICs: Channel Gains May Help}
In \cite{Etkin-Ordentlich}, it is shown that the total DOF of a $K$-user GIC interference channel can be achieved
almost surely when all the cross links have rational gains while the direct links have irrational gains. This result
is generalized by introducing the concept of rational dimensions. The rational dimension of a set of numbers is
defined as the dimension of numbers over the filed of rational numbers. For example, if all numbers are rational then
the dimension is one. We show that if the cross links arriving at a receiver has rational dimension $m$ or less and
it is the case for all receivers then the total DOF of $\frac{K}{m+1}$ is achievable. In special case where $m=1$, it
collapses to the result of Etkin and Ordentlich.

\subsection{Strange Behavior: Discontinuity of DOF}
To highlight some important features of the three-user GIC, the symmetric case in which the channel is governed by a
single channel gain is considered. First, it is proved that when the channel gain is irrational then the total DOF of
the channel can be achieved. This is obtained by using  multi-layer constellations in encoding  together with
Hurwitz's theorem in analysis. There is, however, a subtle difference between this result and the one obtained for
the $K$-user GIC. Here, we prove that the result holds for all irrational numbers while in the $K$-user case we prove
that it holds for almost all real numbers. In fact, there may be some irrational numbers not satisfying the
requirements of the $K$-user case.

When the channel gain is rational then more sophisticated multi-layer constellation design is required to achieve
higher performance. The reason is that interference and data  are sharing the same dimension and splitting them
requires more structure in constellations. We propose a multi-layer constellation in which besides satisfying the
requirement of splitting interference and data, points are packed efficiently in the real line. This is accomplished
by allowing carry over from different levels. Being much simpler in design, avoiding carry over, however, results in
lower DOF. We show that the DOF is roughly related to the maximum of numerator and denominator. But it is always less
than $\frac{3}{2}$.

Viewing the total DOF of the channel as a function of the channel gain, we observe that this function is everywhere
discontinuous which mean its is discontinuous at all points. This is a strange behavior as in all previous results
the DOF is a continuous function almost everywhere. Although this is only achievable, the result of Etkin an
Ordentlich in \cite{Etkin-Ordentlich} confirms that this is in fact the case.

\section{Coding Scheme}\label{sec coding}
In this section, a coding scheme for data transmission in a shared medium is proposed. It is assumed that the
channel is real, additive, and time invariant. The Additive White Gaussian Noise (AWGN)  with variance $\sigma^2$ is
added to the received signals at all receivers. Moreover, transmitters are subject to the power constraint $P$. The
Signal to Noise Ratio (SNR) is defined as $\text{SNR}=\frac{P}{\sigma^2}$.

The proposed coding is rather general and can be applied to several communication systems as it will be explored in
detail in the following sections. In what follows, the encoding and decoding parts of the scheme are explained. The
important features unique to the scheme are also investigated.

\subsection{Encoding}

A transmitter limits its input symbols to a finite set which is called the transmit constellation. Even though
it has access to the continuum of real numbers, restriction to a finite set has the benefit of easy and feasible
interference management. Having a set of finite points as input symbols, however, does not rule out transmission of
multiple data streams from a single transmitter. In fact, there are situations where a transmitter wishes to send
data to several receivers (such as the $X$ channel) or having multiple data streams intended for a single receiver
increases the throughput of the system (such as the interference channel). In what follows, it is shown how a finite
set of points can accommodate different data streams.

Let us first explain the encoding of a single data stream. The transmitter selects a constellation $\mathcal{U}_i$ to
send the data stream $i$. The constellation points are chosen from integer points, i.e., $\mathcal{U}_i\subset
\mathbb{Z}$. It is assumed that $\mathcal{U}_i$ is a bounded set. Hence, there is a constant $Q_i$ such
that $\mathcal{U}_i\subset [-Q_i,Q_i]$. The cardinality of $\mathcal{U}_i$ which limits the rate of data stream $i$
is denoted by $|\mathcal{U}_i|$.

Two choices for the constellation $\mathcal{U}_i$ are considered. The first one, referred to as Type I or single
layer constellation, corresponds to the case where all integers between $-Q_i$ and $Q_i$ are selected. This is a
simple choice yet capable of achieving the total DOF of several channels.

In the second one, referred to as Type II or multi-layer constellation, constellation points are represented to
a base $W\in\mathds{N}$. In other words, a point in the constellation can be
written as
\begin{equation}\label{TypeII}
u_i(\mathbf{b})=\sum_{k=0}^{L-1}{b_{l}W^{l}},
\end{equation}
where $b_{l}\in\{0,1,\ldots,a-1\}$ and
$l\in \{1,2,\ldots,L-1\}$. $\mathbf{b}=(b_{0},\ldots,b_{L-1})$ is in fact another way of expressing
$u_i$ in $W$-array representation. $a$ is the upper limit on the
digits and clearly $a<W$. In fact, if $a=W$ then Type II constellation renders itself as Type I constellation which
is not of interest. Each constellation point can be expressed by $L$ digits and each digit carries independent
message. Each of these digits is referred to as a layer of data. In other words, Type II constellation carries $L$
layers of information.

Having formed the constellation, the transmitter constructs a random codebook for data stream $i$ with rate $R_i$.
This can be accomplished by choosing a probability distribution on the input alphabets. The uniform distribution is
the first candidate and it is selected for the sake of brevity.

tight in general, using this bound does not decrease the performance of the system as long as the DOF is concerned.

In general, the transmitter wishes to send $L$ data streams to one or several receivers. It first constructs $L$ data
streams using the above procedure. Then, it combines them using a linear combination of all data streams. The
transmit signal can be represented by
\begin{equation}
 u=T_1u_1+T_2u_2+\ldots+T_Lu_L,
\end{equation}
where $u_i\in \mathcal{U}_i$ carries information for data stream $i$. $T_i$ is a constant real number that functions
as a separator splitting data stream $i$ from the transmit signal. In fact, one can make an analogy between single
and multiple antenna systems by regarding that the data stream $i$ is in fact transmitted in the direction $T_i$.

$T_i$'s are rationally independent, i.e., the equation $T_1x_1+T_2x_2+\ldots+T_Lx_L=0$
has no rational solutions. This independence is due to the fact that a unique map from constellation points to the
message sets is required. By relying on this independence, any real number $u$ belonging to the set of
constellation points is uniquely decomposable as $u=\sum_{i=1}^L T_i u_i$. Observe that if there is another possible
decomposition $u=\sum_{i=1}^L T_i u_i'$ then it forces $T_i$'s to be dependent.

To adjust the power, the transmitter multiplies the signal by a constant $A$, i.e., the transmit signal is $x=Au$.


\subsection{Received Signal and Interference Alignment}
A receiver in the system may observe a signal which is a linear combination of several data streams and AWGN. The
received signal in its general form can be represented as
\begin{equation}
 y=g_0 u_0+\underbrace{g_1 u_1+\ldots+g_M u_M}_I +z,
\end{equation}
where $u_i$ is the received signal corresponding to the data stream $i$ and $z$ is the AWGN with covariance
$\sigma^2$. $g_i$ is a constant which encapsulates several multiplicative factors from a transmitter to the
receiver. Without loss of generality, it is assumed that the receiver wishes to decode the first data stream
$u_0$ which is encoded with rate $R_0$. The rest of data streams is the interference for the intended data stream and
is denoted by $I$.

The proposed encoding scheme is not optimal in general. However, it provides sufficient tools to accomplish
interference alignment in the network which in turn maximizes the throughput of the system. In $n$-dimensional
Euclidean spaces ($n\geq2$), two interfering signals are aligned when they receive in the same direction at the
receiver. In general, $m$ signals are aligned at a receiver if they span a subspace with dimension less than $m$. We
claim that, surprisingly,  similar arguments can be applied in one-dimensional spaces. The definition of aligned data
streams is
needed first.
\begin{definition}[Aligned Data Streams]
Two data streams $u_i$ and $u_j$ are said to be aligned at a receiver if the receiver observes a rational
combination of them.
\end{definition}

As it will be shown in the following sections, if two streams are aligned then their effect at the
receiver is similar to a single data stream at high SNR regimes. This is due to the fact that rational numbers
form a filed and therefore the sum of constellations is again a constellation from $\mathbb{Q}$ with
enlarged cardinality.

To increase $R_0$, it is desirable to align data streams in the interference
part of the signal, i.e. $I$. The interference alignment in its simplest form happens when several
data streams arrive at the receiver with similar coefficients, e.g. $I=gu_1+gu_2+\ldots+gu_M$. In this case, the
data streams can be bundled to a single stream with the same coefficient. It is possible to extend this simple case
of interference alignment to more general cases. First, the following definition is needed.
\begin{definition}[Rational Dimension]
The rational dimension of a set of real numbers $\{h_1,h_2,\ldots,h_M\}$ is $m$ if there exists a set of real
numbers $\{H_1,H_2,\ldots,H_m\}$ such that each $h_i$ can be represented as a rational combination of
$H_j$'s, i.e., $h_i=\alpha_{i1}H_1+\alpha_{i2}H_2+\ldots+\alpha_{im}H_m$ where $\alpha_{ik}\in \mathbb{Q}$ for
all $k\in\{1,2,\ldots,m\}$. In particular, $\{h_1,h_2,\ldots,h_M\}$ are rationally independent if the
rational dimension is $M$, i.e., none of the numbers can be represented as the rational combination of other numbers.
\end{definition}

\begin{remark}
In the above definition, one can replace the set of rational numbers with integers as multiplication of irrational
numbers with integers results in irrational numbers. Therefore, the two alternative definitions are used in this
paper.
\end{remark}

In fact, the rational dimension is the effective dimension seen at the receiver. To see this, suppose that the
coefficients in the interference part of the signal $I=g_1u_1+g_2u_2+\ldots+g_Mu_M$ has rational
dimension $m$ with bases $\{G_1,\ldots,G_m\}$. Therefore, each $g_i$ for $i\in\{1,2,\ldots,M\}$ can be written as
$g_i=\alpha_{i1}G_1+\alpha_{i2}G_2+\ldots+\alpha_{im}G_m$ where $\alpha_{ik}$ is an integer. Plugging into the
equation, it is easy to see that $I$ can be represented as $I=G_1I_1+G_2I_2+\ldots+G_mI_m$ where $I_k$ is a linear
combination of data streams with integer coefficient. In fact, if the coefficients have dimension $m$ then the
interference part of the signal occupies $m$ rational dimensions and one dimension is available for the signal. On
the
other hand, since the dimension is one, it can be concluded that multiplexing gain of the intended data stream is
$\frac{1}{m+1}$. In one extreme case the rational dimension is one and all coefficients are an integer
multiple of a real number and $m=1$.

\subsection{Decoding}
After rearranging the interference part of the signal, the received signal can be represented as
\begin{equation}\label{received signal}
 y=G_0u_0+G_1I_1+\ldots+G_mI_m+z,
\end{equation}
where $G_0=g_0$ to unify the notation. In what follows, the decoding scheme used to decode $u_1$ from $y$ is
explained. It is worth noting that if the receiver is interested in more than one data stream then it performs the
same decoding procedure for each data stream.

At the receiver, the received signal is first passed through a hard decoder.
The hard decoder looks at the received constellation
$\mathcal{U}_r=G_0\mathcal{U}_0+G_1\mathcal{I}_1+\ldots+G_m\mathcal{I}_m$ and maps the received
signal to the nearest point in the constellation. This changes the continuous channel to a discrete one in which the
input symbols are from the transmit constellation $\mathcal{U}_1$ and the output symbols are from received
constellation.

\begin{remark}
$\mathcal{I}_j$ is the constellation due to single or multiple data streams. Since it is assumed that in the
latter case it is a linear combination of multiple data streams with integer coefficients, it can be concluded that
$\mathcal{I}_j\subset \mathbb{Z}$ for $j\in\{1,2,\ldots,m\}$.
\end{remark}

To bound the performance of the decoder, it is assumed that the received constellation has the property that there
is a many-to-one map from $\mathcal{U}_r$ to $\mathcal{U}_0$. This in fact implies that if there is no additive
noise in the channel then the receiver can decode the data stream with zero error probability. This property is
called property $\Gamma$. It is assumed that this property holds for all received constellations. To satisfy this
requirement at all receivers, usually a careful transmit constellation design is needed at all transmitters.

Let $d_{\text{min}}$ denote the minimum distance in the received constellation. Having Property $\Gamma$, the
receiver passes the output of the hard decoder through the many-to-one map from $\mathcal{U}_r$ to $\mathcal{U}_0$.
The output is called $\hat{u}_1$. Now, a joint-typical decoder can be used to decode the data stream from a block of
$\hat{u}_0$s. To calculate the achievable rate of this scheme, the error probability of transmitting a symbol from
$\mathcal{U}_0$ and receiving another symbol, i.e. $P_e=Pr\{\hat{U}_0\neq U_0\}$ is bounded as:
\begin{IEEEeqnarray}{rl}\label{error probability}
P_e &\leq Q\left(\frac{d_{\text{min}}}{2\sigma}\right)\leq
\exp\left({-\frac{ d_{\text{min}}^2}{8\sigma^2}}\right).
\end{IEEEeqnarray}

Now, $P_e$ can be used to lower bound the rate achievable for the data stream. In \cite{Etkin-Ordentlich}, Etkin and
Ordentlich used Fano's
inequality to obtain a lower bound on the achievable rate which is tight in high SNR regimes. Following similar
steps, one can obtain
\begin{IEEEeqnarray}{rl}
 R_0 & = H(\hat{U}_0,U_0)\nonumber\\
     & =H(U_0)-H(U_0|\hat{U}_0)\nonumber\\
     & \stackrel{a}{\geq} H(U_0)-1-P_e\log |\mathcal{U}_0|\nonumber\\
     & \stackrel{b}{\geq} \log |\mathcal{U}_0|-1-P_e\log |\mathcal{U}_0| \label{lower bound on R}
\end{IEEEeqnarray}
where (a) follows from Fano's inequality and (b) follows from the fact that $U_1$ has the uniform distribution. To
have multiplexing gain of at least $r_0$, $|U_1|$ needs to scale as $SNR^{r_0}$. Moreover, if $P_e$ scales as
$\exp\left(SNR^{-\epsilon}\right)$ for an $\epsilon>0$ then it can be shown that $\frac{R_0}{\log \text{SNR}}$
approaches $r_0$ at high SNR regimes.

\begin{remark}
After interference alignment the interference term has no longer the uniform distribution. However, the lower bound
on the achievable rate given in (\ref{lower bound on R}) is independent of the probability distributions of the
interference terms. It is possible to obtain better performance provided the distribution of the interference is
exploited.
\end{remark}

\section{Single Layer Constellation}
In this section, the single layer constellation is used to modulate all data streams at all transmitters. Even
though it is the simplest form of constellation, it is powerful enough to provide interference alignment which in
turn increases the throughput of the system. Before deriving import results regarding DOF of the $X$ and interference
channels using this constellation, the performance of a typical decoder is analyzed. The attempt is to make the
analysis universal and applicable to both channels.

\subsection{Peformance Analysis: The Khintchine-Groshev Theorem}
The decoding scheme proposed in the previous section is
used to decode the data stream $u_0$ from the received signal in (\ref{received signal}). To satisfy Property
$\Gamma$, it is assumed that $\{G_0,G_1,\ldots,G_m\}$ are independent over rational numbers. Due to this
independence, any point in the received constellation has a unique representation in the bases
$\{G_0,G_1,\ldots,G_m\}$ and therefore Property $\Gamma$ holds in this case.

\begin{remark}
In a random environment, it is easy to show that the set of $\{G_0,G_1,\ldots,G_m\}$ being dependent
has measure zero (with respect to Lebesgue measure). Hence, in this section it is assumed that Property $\Gamma$
holds unless otherwise stated.
\end{remark}

To use the lower bound on the data rate given in (\ref{lower bound on R}), one needs to calculate the minimum
distance between points in the received constellation. Let us assume each stream in (\ref{received signal}) is
bounded (as it is the case since transmit constellations are bounded by the assumption). In particular,
$\mathcal{U}_0=[-Q_0,Q_0]$ and $\mathcal{I}_j=[-Q_j,Q_j]$ for all $j\in\{1,2,\ldots,m\}$.  Since points in the
received constellation are irregular, finding $d_{\min}$ is not easy in general. Thanks to the theorems of Khintchine
and Groshev, however, it is possible to lower bound the minimum distance. As it will be shown later, using this lower
bound at high SNR regimes is asymptotically optimum. We digress here and explain some background
needed for stating the theorem of Khintchine and Groshev.

The field of Diophantine approximation in number theory deals with approximation of real numbers with rational
numbers. The reader is referred to \cite{schmidt,hardy} and the references therein. The Khintchine theorem is one of
the cornerstones in this field. It gives a criteria for a given function $\psi:\mathbb{N}\to\mathbb{R}_+$ and real
number $\alpha$ such that $|p+\alpha q|<\psi(|q|)$  has
either infinitely many solutions or at most finitely many solutions for $(p,q)\in \mathbb{Z}^2$. Let
$\mathcal{A}(\psi)$ denote the set of real numbers such that $|p+\alpha q|<\psi(|q|)$  has infinitely many
solutions in integers. The theorem has two parts. The first part is the convergent part and states that if
$\psi(|q|)$ is convergent, i.e., $$ \sum_{q=1}^\infty \psi(q)<\infty$$ then $\mathcal{A}(\psi)$ has measure zero
with respect to Lebesque measure. This part can be rephrased in more convenient way as follows. For almost all real
numbers, $|p+\alpha q|>\psi(|q|)$ holds for all $(p,q)\in \mathbb{Z}^2$ except for finitely many of them. Since the
number of integers violating the inequality is finite, one can find a constant $\kappa$ such that  $$|p+\alpha
q|>\kappa\psi(|q|)$$ holds for all integers $p$ and $q$ almost surely. The divergent part of the theorem states that
$\mathcal{A}(\psi)$ has the full measure, i.e. the set $\mathbb{R}-\mathcal{A}(\psi)$ has measure zero, provided
$\psi$ is decreasing and $\psi(|q|)$ is divergent, i.e., $$ \sum_{q=1}^\infty \psi(q)=\infty.$$

There is an extension to Khintchine's theorem which regards the approximation of linear forms. Let
$\boldsymbol{\alpha}=(\alpha_1,\alpha_2,\ldots,\alpha_m)$ and $\mathbf{q}=(q_1,q_2,\ldots,q_m)$ denote an $m$-tuple
in $\mathbb{R}^m$ and $\mathbb{Z}^m$, respectively. Let $\mathcal{A}_m(\psi)$ denote the set of $m$-tuple real
numbers $\boldsymbol{\alpha}$ such that
\begin{equation}
 |p+\alpha_1q_1+\alpha_2q_2+\ldots+\alpha_mq_m|<\psi (|\mathbf{q}|_{\infty})
\end{equation}
has infinitely many solutions for $p\in \mathbb{Z}$ and $\mathbf{q}\in\mathbb{Z}^m$. $|\mathbf{q}|_{\infty}$
is the supreme norm of $\mathbf{q}$ defined as $\max_i |q_i|$. The following theorem gives the Lebesque measure of
the set $\mathcal{A}_m(\psi)$.

\begin{theorem}[Khintchine-Groshev]
Let $\psi:\mathbb{N}\to\mathbb{R}^+$. Then the set $\mathcal{A}_{m}(\psi)$ has measure zero provided
\begin{equation}\label{convergence}
 \sum_{q=1}^\infty q^{m-1}\psi(q)<\infty,
\end{equation}
and has the full measure if
\begin{equation}
 \sum_{q=1}^\infty q^{m-1}\psi(q)=\infty \quad  \text{ and $\psi$ is monotonic}.
\end{equation}
\end{theorem}

In this paper, the convergent part of the theorem is concerned. Moreover, given an arbitrary
$\epsilon>0$ the function $\psi(q)=\frac{1}{q^{m+\epsilon}}$ satisfies (\ref{convergence}). In fact, the convergent
part of the theorem used in this paper can be stated as follows. For almost all $m$-tuple real numbers
there exists a constant $\kappa$ such that
\begin{equation}\label{khintchine}
 |p+\alpha_1q_1+\alpha_2q_2+\ldots+\alpha_mq_m|>\frac{\kappa}{(\max_i|q_i|)^{m+\epsilon}}
\end{equation}
holds for all $p\in \mathbb{Z}$ and $\mathbf{q}\in\mathbb{Z}^m$.

The Khintchine-Groshev theorem can be used to bound the minimum distance of points in the received constellation. In
fact, a point in the received constellation has a linear form, i.e., $u_r=G_0u_0+G_1I_1+\ldots+G_mI_m$. Dividing by
$G_0$ and using (\ref{khintchine}), one can conclude that
\begin{equation}\label{khintchine-mg}
 d_{\min}>\frac{\kappa G_0}{(\max_{i\in\{1,\ldots,m\}}Q_i)^{m+\epsilon}}
\end{equation}

The probability of error in hard decoding, see (\ref{error probability}), can be bounded as
\begin{equation}\label{alaki1}
P_e<\exp\left({-\frac{ (\kappa G_0)^2}{8\sigma^2(\max_{i\in\{1,\ldots,m\}}Q_i)^{2m+2\epsilon}}}\right).
\end{equation}

Let us assume $Q_i$ for $i\in\{0,1,\ldots,m\}$ is  $\lfloor  \gamma_i
P^{\frac{1-\epsilon}{2(m+1+\epsilon)}}\rfloor$ where $\gamma_i$ is a constant. Moreover, $\epsilon$ is the
constant appeared in (\ref{khintchine}). We also assume that $G_0=\gamma P^{\frac{m+2\epsilon}{2(m+1+\epsilon)}}$.
As it will be shown later, these
assumptions are realistic and can be applied to the coding schemes proposed in this paper. It is worth
mentioning that in this paper it is assumed that each data stream carries the same rate in the asymptotic case of
high
SNR, i.e., they have the same multiplexing gain. However, in more general cases one may consider different
multiplexing gains
for different data streams. Substituting in (\ref{alaki1}) yields
\begin{equation}
P_e<\exp\left(- \delta P^{\epsilon}\right),
\end{equation}
where $\delta$ is a constant and a function of $\gamma$, $\kappa$, $\sigma$, and $\gamma_i$'s. The lower bound
obtained in (\ref{lower bound on R}) for the achievable rate becomes
\begin{IEEEeqnarray}{rl}
 R_0  & > (1-P_e)\log |\mathcal{U}_0|-1 \nonumber\\
      & \stackrel{a}{=} \left(1-\exp\left(- \delta P^{\epsilon}\right)\right)\log(2\lfloor  \gamma_i
P^{\frac{1-\epsilon}{2(m+1+\epsilon)}}\rfloor)-1\nonumber\\
      & > \frac{(1-\epsilon)\left(1-\exp\left(- \delta P^{\epsilon}\right)\right)}{2(m+1+\epsilon)}(\log(P)
+\vartheta)-1 \label{alaki2}
\end{IEEEeqnarray}
where (a) follows from the fact that $|\mathcal{U_0}|=2Q_0$ and $\vartheta$ is a constant. The multiplexing gain of
the data stream $u_0$ can be computed using (\ref{alaki2}) as follows
\begin{IEEEeqnarray}{rl}
 r_0 &=\lim_{P\rightarrow\infty} \frac{R_0}{0.5\log(P)}\nonumber\\
     & > \frac{1-\epsilon}{m+1+\epsilon}.
\end{IEEEeqnarray}
Since $\epsilon$ can be made arbitrarily small, we can conclude that $r=\frac{1}{m+1}$ is indeed achievable. In the
following theorem, this result and its required conditions are summarized.

\begin{theorem}\label{basic}
A receiver can reliably decode the data stream $u_0$ with multiplexing gain $\frac{1}{m+1}$ from the received signal
$y=G_0u_0+G_1I_1+\ldots+G_mI_m+z$ if the following regularity conditions are satisfied:

\begin{enumerate}
\item $G_0=\gamma P^{\frac{m+2\epsilon}{2(m+1+\epsilon)}}$ where $\gamma$ is a consant.
\item $u_0\in [-Q_0,Q_0] $ where $Q_0=\lfloor  \gamma_0
P^{\frac{1-\epsilon}{2(m+1+\epsilon)}}\rfloor$ and $\gamma_0$ is a constant. Moreover, the uniform distribution is
used to construct the random codebook.
\item For $i\in\{1,2,\ldots,m\}$, $I_i\in [-Q_i,Q_i] $ where $Q_i=\lfloor  \gamma_i
P^{\frac{1-\epsilon}{2(m+1+\epsilon)}}\rfloor$ and $\gamma_i$ is a constant.
\item $G_i$s for $i\in\{0,1,\ldots,_m\}$ are independent over rational numbers.
\item $\{\frac{G_1}{G_0},\frac{G_2}{G_0},\ldots,\frac{G_m}{G_0}\}$ is among $m$-tuples that satisfy
(\ref{khintchine}).
\end{enumerate}
Moreover, the last two conditions hold almost surely.
\end{theorem}


\subsection{Two-user $X$ channel: $\text{DOF}=\frac{4}{3}$ is Achievable Almost Surely}

\begin{figure}
\centering
\scalebox{.75} 
{
\begin{pspicture}(0,-1.718125)(9.902813,1.718125)
\rput{-180.0}(6.881875,2.739375){\pstriangle[linewidth=0.04,dimen=outer](3.4409375,1.2396874)(0.6,0.26)}
\psline[linewidth=0.04cm](3.4409375,1.2596875)(3.4409375,0.6596875)
\psline[linewidth=0.04cm](3.2209375,0.6596875)(3.6609375,0.6596875)
\rput{-180.0}(6.881875,-1.260625){\pstriangle[linewidth=0.04,dimen=outer](3.4409375,-0.7603125)(0.6,0.26)}
\psline[linewidth=0.04cm](3.4409375,-0.7403125)(3.4409375,-1.3403125)
\psline[linewidth=0.04cm](3.2209375,-1.3403125)(3.6609375,-1.3403125)
\rput{-180.0}(12.481875,2.739375){\pstriangle[linewidth=0.04,dimen=outer](6.2409377,1.2396874)(0.6,0.26)}
\psline[linewidth=0.04cm](6.2409377,1.2596875)(6.2409377,0.6596875)
\psline[linewidth=0.04cm](6.0209374,0.6596875)(6.4609375,0.6596875)
\rput{-180.0}(12.481875,-1.260625){\pstriangle[linewidth=0.04,dimen=outer](6.2409377,-0.7603125)(0.6,0.26)}
\psline[linewidth=0.04cm](6.2409377,-0.7403125)(6.2409377,-1.3403125)
\psline[linewidth=0.04cm](6.0209374,-1.3403125)(6.4609375,-1.3403125)

\rput(3.8009374,0.9796875){\rnode{T1}{}}
\rput(3.8009374,-1.0203125){\rnode{T2}{}}
\rput(6.0009375,0.9796875){\rnode{R1}{}}
\rput(6.0009375,-1.0203125){\rnode{R2}{}}
\ncline{->}{T1}{R1}
\ncput*[npos=.5]{\small $h_{11}$}
\ncline{->}{T1}{R2}
\ncput*[nrot=:U,npos=.7]{\small $h_{21}$}
\ncline{->}{T2}{R1}
\ncput*[nrot=:U,npos=.7]{\small $h_{12}$}
\ncline{->}{T2}{R2}
\ncput*[npos=.5]{\small $h_{22}$}

\psframe[linewidth=0.04,dimen=outer](2.6009376,1.3796875)(1.0009375,0.5796875)
\psline[linewidth=0.04cm,arrowsize=0.05291667cm
2.0,arrowlength=1.4,arrowinset=0.4]{->}(2.5809374,1.0196875)(3.2209375,1.0196875)
\psline[linewidth=0.04cm,arrowsize=0.05291667cm
2.0,arrowlength=1.4,arrowinset=0.4]{->}(0.3809375,1.1996875)(1.0209374,1.1796875)
\psline[linewidth=0.04cm,arrowsize=0.05291667cm
2.0,arrowlength=1.4,arrowinset=0.4]{->}(0.3809375,0.7996875)(1.0209374,0.7796875)

\usefont{T1}{ptm}{m}{n}
\rput(0.45234376,1.4246875){\textcolor{blue}{$U_1$}}
\rput(0.48234376,0.5646875){\textcolor{red}{$V_1$}}
\rput(3.0023437,1.3246875){$x_1$}
\rput(9.352344,1.5246875){\textcolor{blue}{$\hat{U}_1$}}
\rput(9.372344,0.5246875){\textcolor{blue}{$\hat{U}_2$}}
\rput(6.952344,1.3246875){$y_1$}
\rput(1.8182813,1.0046875){Encoder}
\rput(1.8582813,-1.0153126){Encoder}
\rput(0.49234375,-0.5953125){\textcolor{blue}{$U_2$}}
\rput(0.52234375,-1.4553125){\textcolor{red}{$V_2$}}
\rput(2.9823437,-0.6953125){$x_2$}
\rput(9.322344,-0.4953125){\textcolor{red}{$\hat{V}_1$}}
\rput(9.342343,-1.4953125){\textcolor{red}{$\hat{V}_2$}}
\rput(6.932344,-0.6953125){$y_2$}

\psframe[linewidth=0.04,dimen=outer](8.840938,1.4196875)(7.2409377,0.6196875)
\usefont{T1}{ptm}{m}{n}
\rput(8.068906,1.0446875){Decoder}
\psline[linewidth=0.04cm,arrowsize=0.05291667cm
2.0,arrowlength=1.4,arrowinset=0.4]{->}(6.6209373,0.9996875)(7.2609377,0.9996875)
\psline[linewidth=0.04cm,arrowsize=0.05291667cm
2.0,arrowlength=1.4,arrowinset=0.4]{->}(8.840938,1.2196875)(9.480938,1.1996875)
\psline[linewidth=0.04cm,arrowsize=0.05291667cm
2.0,arrowlength=1.4,arrowinset=0.4]{->}(8.840938,0.8196875)(9.480938,0.7996875)

\psframe[linewidth=0.04,dimen=outer](2.6409376,-0.6403125)(1.0409375,-1.4403125)
\psline[linewidth=0.04cm,arrowsize=0.05291667cm
2.0,arrowlength=1.4,arrowinset=0.4]{->}(2.6209376,-1.0003124)(3.2609375,-1.0003124)
\psline[linewidth=0.04cm,arrowsize=0.05291667cm
2.0,arrowlength=1.4,arrowinset=0.4]{->}(0.4209375,-0.8203125)(1.0609375,-0.8403125)
\psline[linewidth=0.04cm,arrowsize=0.05291667cm
2.0,arrowlength=1.4,arrowinset=0.4]{->}(0.4209375,-1.2203125)(1.0609375,-1.2403125)
\psframe[linewidth=0.04,dimen=outer](8.820937,-0.6003125)(7.2209377,-1.4003125)
\usefont{T1}{ptm}{m}{n}
\rput(8.048906,-0.9753125){Decoder}
\psline[linewidth=0.04cm,arrowsize=0.05291667cm
2.0,arrowlength=1.4,arrowinset=0.4]{->}(6.6009374,-1.0203125)(7.2409377,-1.0203125)
\psline[linewidth=0.04cm,arrowsize=0.05291667cm
2.0,arrowlength=1.4,arrowinset=0.4]{->}(8.820937,-0.8003125)(9.4609375,-0.8203125)
\psline[linewidth=0.04cm,arrowsize=0.05291667cm
2.0,arrowlength=1.4,arrowinset=0.4]{->}(8.820937,-1.2003125)(9.4609375,-1.2203125)
\end{pspicture}
}
\caption{The two-user $X$ channel: Transmitter 1 sends data streams $U_1$ and $V_1$ to Receiver 1 and 2,
respectively.
Similarly, Transmitter 2 sends data streams $U_2$ and $V_2$ to Receiver 1 and 2, respectively.}\label{xchannelfigure}
\end{figure}
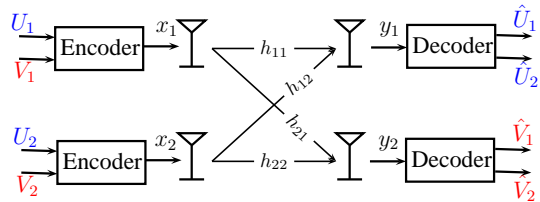

The proposed coding scheme using the single layer constellation is applied to the two-user $X$ channel as the first
example. The two-user $X$ channel is introduced in \cite{maddahali2008com} where the first explicit interference
alignment is used to achieve the total DOF of a class of  MIMO $X$ channels. In this channel, see
Figure \ref{xchannelfigure}, there are two transmitters and two receivers. Transmitter 1 wishes to send data streams
$U_1$ and $V_1$ to Receivers 1 and 2, respectively. Similarly, Transmitter 1 wishes to send data streams $U_2$
and $V_2$ to Receivers 1 and 2, respectively. The input-output relation of the channel can be stated as
\begin{IEEEeqnarray}{rl}
 y_1 &=h_{11}x_1+h_{12}x_2+z_1,\nonumber\\
 y_2 &=h_{21}x_1+h_{22}x_2+z_2,\nonumber
\end{IEEEeqnarray}
where $z_1$ and $z_2$ are AWGN with variance $\sigma^2$. $x_1$ and $x_2$ are input symbols of Transmitter 1 and 2,
respectively. Input signals are subject to the power constraint $P$. $h_{ij}$ is the channel gain from Transmitter
$j$ to Receiver $i$. Moreover, channel gains are assumed to be constant over time. $y_1$ and $y_2$ are received
signals at Receiver 1 and 2, respectively.

In \cite{jafar2008dfr}, an upper bound on the DOF of the channel is obtained. This upper bound for the single
antenna case is $\frac{4}{3}$. We will show that this upper bound is in fact achievable. If each data stream
occupies $\frac{1}{3}$ of DOF then the total DOF becomes $\frac{4}{3}$. Therefore, it is assumed that all data
streams, i.e. $U_1$, $U_2$, $V_1$ and $V_2$, use the same constellation with integer points from interval $[-Q,Q]$
with $Q=\lfloor \gamma
P^{\frac{1-\epsilon}{2(3+\epsilon)}}\rfloor$ where $\gamma$ and $\epsilon$ are two arbitrary constants. Transmitter
1 (respectively 2) encodes the data streams $U_1$ and $V_1$ (respectively $U_2$ and $V_2$) utilizing the encoding
scheme proposed in the previous section. The following linear combinations are used to send the data streams through
the channel.
\begin{IEEEeqnarray}{rl}
x_1=G(h_{22}u_1+h_{12}v_1),\label{alaki3}\\
x_2=G(h_{21}u_2+h_{11}v_2),\label{alaki4}
\end{IEEEeqnarray}
where $G$ is the normalizing factor. To find $G$, one needs to calculate the transmit power of User 1 and 2. It is
easy to show that there exists a constant $\gamma'$ such that $G=\gamma' P^{\frac{2+2\epsilon}{2(3+\epsilon)}}$
normalizes the transmit power to be less than $P$ at both receivers.

After rearranging, the received signal can be written as
\begin{IEEEeqnarray}{rl}
 y_1
&=Gh_{11}h_{22}u_1+Gh_{12}h_{21}u_2+Gh_{11}h_{12}(\underbrace{
v_1+v_2}_{I_1})+z_1,\nonumber\\
 y_2
&=Gh_{21}h_{22}(\underbrace{u_1+u_2}_{I_2})+Gh_{12}h_{21}v_1+Gh_{11}h_{22}v_2 +z_2.\nonumber
\end{IEEEeqnarray}
Now, it becomes clear why the linear combinations in (\ref{alaki3}) and (\ref{alaki4}) are used to combine the data
streams at the transmitters. In fact,  the data streams $V_1$ and $V_2$ not intended for Receiver 1 arrive with the
same coefficients at Receiver 1. In other words, they are aligned at the receiver and hence their effect can be
regarded as a single data stream. Let $I_1$ denote the sum $v_1+v_2$. Clearly, $I_1$ is an integer and belongs to
$[-2Q\ 2Q]$. Receiver 1 wishes to decode $U_1$ and $U_2$. As proposed in the previous section, each data stream is
decode separately at the receiver. Therefore, decoding of the data stream $U_1$ is first considered. It is easy to
see that all regularity conditions given in Theorem \ref{basic} are satisfied with $m=2$. Hence, Receiver 1 can
reliably decode $U_1$ which has the multiplexing gain of $\frac{1}{3}$. Similarly, Receiver 2 can decode $U_2$ which
has the multiplexing gain of $\frac{1}{3}$. A similar phenomenon happens in the second receiver. Therefore, we have
proved the following theorem.
\begin{theorem}
The DOF of the two-user $X$ channel is $\frac{4}{3}$ almost surely.
\end{theorem}

\subsection{$K$-user Gaussian Interference Channel: Special Cases}

\begin{figure}
 \centering
\scalebox{1} 
{
\begin{pspicture}(0,-2.99)(6.4628124,2.97)
\psset{linewidth=0.03cm,arrowsize=0.05291667cm
2.0,arrowlength=1.4,arrowinset=0.4}
\usefont{T1}{ptm}{m}{n}

\def\antenna{%
\begin{pspicture}(1,1)
\pstriangle[gangle=-180.0](0,0.25)(0.6,0.25)
\psline(0,0)(0,-0.60)
\psline(-0.22,-0.60)(0.22,-0.60)
\end{pspicture}
}

\rput[bl](1.7,2.5){\antenna}
\rput[bl](1.7,0.5){\antenna}
\rput[bl](1.7,-2.6){\antenna}

\rput[bl](5.5,2.5){\antenna}
\rput[bl](5.5,0.5){\antenna}
\rput[bl](5.5,-2.6){\antenna}

\psline{->}(.8,2.45)(1.4,2.45)
\psline{->}(5.7,2.45)(6.3,2.45)
\rput(0.45,2.45){$x_1$}
\rput(6.7,2.45){$y_1$}

\psline{->}(.8,0.45)(1.4,0.45)
\psline{->}(5.7,0.45)(6.3,0.45)
\rput(0.45,0.45){$x_2$}
\rput(6.7,0.45){$y_2$}

\psline{->}(.8,-2.65)(1.4,-2.65)
\psline{->}(5.7,-2.65)(6.3,-2.65)
\rput(0.45,-2.65){$x_K$}
\rput(6.7,-2.65){$y_K$}

\cnode[linecolor=white](2,2.45){.15}{T1}
\cnode[linecolor=white](2,0.45){.15}{T2}
\cnode[linecolor=white](2,-2.65){.15}{Tk}

\cnode[linecolor=white](5.1,2.45){.15}{R1}
\cnode[linecolor=white](5.1,0.45){.15}{R2}
\cnode[linecolor=white](5.1,-2.65){.15}{Rk}

\ncline{->}{T1}{R1}
\ncput*[npos=.8]{\tiny $h_{11}$}
\ncline[linecolor=red]{->}{T1}{R2}
\ncput*[nrot=:U,npos=.85]{\tiny $h_{21}$}
\ncline[linecolor=red]{->}{T2}{R1}
\ncput*[nrot=:U,npos=.7]{\tiny $h_{12}$}
\ncline{->}{T2}{R2}
\ncput*[npos=.8]{\tiny $h_{22}$}
\ncline{->}{Tk}{Rk}
\ncput*[npos=.75]{\tiny $h_{KK}$}
\ncline[linecolor=red]{->}{Tk}{R1}
\ncput*[nrot=:U,npos=.85]{\tiny $h_{1K}$}
\ncline[linecolor=red]{->}{Tk}{R2}
\ncput*[nrot=:U,npos=.8]{\tiny $h_{2K}$}
\ncline[linecolor=red]{->}{T1}{Rk}
\ncput*[nrot=:U,npos=.8]{\tiny $h_{K1}$}
\ncline[linecolor=red]{->}{T2}{Rk}
\ncput*[nrot=:U,npos=.7]{\tiny $h_{K2}$}

\psdots[dotsize=0.12](5.5,-0.69)
\psdots[dotsize=0.12](5.5,-1.13)
\psdots[dotsize=0.12](5.5,-1.51)
\psdots[dotsize=0.12](1.6209375,-0.75)
\psdots[dotsize=0.12](1.6409374,-1.19)
\psdots[dotsize=0.12](1.6209375,-1.57)
\end{pspicture}
}
\caption{The $K$-user GIC. User $i$ for $i\in \{1,2,\ldots,K\}$ wishes to communicate with its
corresponding receiver while receiving interference from other users.}\label{k-user IC}
\end{figure}

The $K$-user GIC models a network in which $K$ transmitter-receiver pairs (users) sharing a common bandwidth wish
to have reliable communication at maximum rate. The channel's input-output relation can be stated as, see Figure
\ref{k-user IC},
\begin{IEEEeqnarray}{rl}\label{k-user model}
 y_1 &=h_{11}x_1+h_{12}x_2+\ldots +h_{1K}x_K+z_1,\nonumber\\
 y_2 &=h_{21}x_1+h_{22}x_2+\ldots +h_{2K}x_K+z_2,\nonumber\\
\vdots\ &=\quad  \vdots \quad\qquad \vdots \quad\qquad\ddots\qquad\vdots \\
y_K &=h_{K1}x_1+h_{K2}x_2+\ldots +h_{KK}x_K+z_K,\nonumber
\end{IEEEeqnarray}
where $x_i$ and $y_i$ are input and output symbols of User $i$ for $i\in\{1,2,\ldots,K\}$, respectively. $z_i$ is
AWGN with variance $\sigma^2$ for $i\in\{1,2,\ldots,K\}$. Transmitters are subject to the power constraint $P$.

An upper bound on the DOF of this channel is obtained in \cite{cadambe2008iaa}. The upper bound states that the
total DOF of the channel is less than $\frac{K}{2}$ which means each user can at most use one half of its maximum
DOF. This upper bound can be achieved by using single layer constellation in special case where all cross gains are
rational numbers \cite{Etkin-Ordentlich}. This is due to the fact that these coefficients lie on a single rational
dimensional space and therefore the effect of the interference caused by several transmitters behaves as that of
interference caused by a single transmitter. Using a single data stream, one can deduce that the multiplexing gain of
$\frac{1}{2}$ is achievable for each user.

Restriction to transmission of single data streams is not optimal in general. As an example showing this fact, in
the next subsection, it is proved that by having multiple data streams one can obtain higher DOF. However, using
single data streams has the advantage of simple analysis. We are interested in the DOF of the system when each user
employs a single data stream. The following theorem states the result. This in fact
generalizes the result obtained in \cite{Etkin-Ordentlich}.

\begin{theorem}\label{theorem-k-user-special-case}
 The DOF of $\frac{K}{m+1}$ is achievable for the $K$-user Gaussian interference channel using the single data stream
transmission scheme provided the set of cross gains at each receiver has the rational dimension of at most $m$.
\end{theorem}
\begin{proof}
To communicate with its corresponding receiver, each transmitter transmits one data stream modulated with single
layer constellation. It is assumed that all users use the same constellation, i.e., $\mathcal{U}_i=[-Q\ Q]$ for
$i\in\{1,2,\ldots,K\}$. We claim that under the conditions assumed in the theorem each transmitter can achieve the
multiplexing gain of $\frac{1}{m+1}$. To accommodate this data rate, $Q$ is set to $\lfloor
P^{\frac{1-\epsilon}{2(m+1+\epsilon)}}\rfloor$. The transmit signal from Transmitter $i$ is $x_i=Gu_i$ for
$i\in\{1,2,\ldots,k\}$ where $G$ is the normalizing factor and equals $\gamma
P^{\frac{m+2\epsilon}{2(m+1+\epsilon)}}$ and $\gamma$ is a constant. Due to the symmetry obtained by proposed coding
scheme, it is sufficient to analyze the performance of the first user. The received signal at Receiver 1 can be
represented as
\begin{equation}\label{alaki5}
y_1=G(h_{11}u_1+h_{1K}u_2+\ldots+h_{1K}u_K)+z_1.
\end{equation}
Let us assume the rational dimension of $(h_{12},h_{13},\ldots, h_{1K})$ is less than $m$. Hence, there exists a set
of real numbers $(g_1,g_2,\ldots,g_m)$ such that each $h_{1j}$ can be represented as
\begin{equation}
h_{1j}=\sum_{l=1}^{m}\alpha_{jl}g_l,
\end{equation}
where $\alpha_{jl}\in\mathbb{Z}$ for $j\in\{2,\ldots,K\}$ and $l\in\{1,2,\ldots,m\}$. Substituting in (\ref{alaki5})
and rearranging yields
\begin{equation}\label{alaki6}
y_1=G(h_{11}u_1+g_1 I_1+\ldots+g_mI_m)+z_1.
\end{equation}
where $I_l\in\mathbb{Z}$ for $l\in\{1,2,\ldots,m\}$ and
\begin{equation}
I_l =\sum_{j=2}^{K}\alpha_{jl} u_j.
\end{equation}
It is easy to prove that there is a constant $\gamma_l$ such that $I_l\in[-Q_l\ Q_l]$ for $l\in\{1,2,\ldots,m\}$ 
where  $Q_l=\lfloor \gamma_l P^{\frac{1-\epsilon}{2(m+1+\epsilon)}}\rfloor$. Receiver 1 decodes its corresponding
data stream from received signal in (\ref{alaki6}) using the decoding scheme proposed in the previous section. By
one-to-one correspondence with regularity conditions in Theorem \ref{basic}, one can deduce that Receiver one is able
to decode the data stream $u_1$ and in fact the multiplexing gain of $\frac{1}{m+1}$ is achievable almost surely. Due
to the symmetry, we can conclude that the DOF of $\frac{K}{m+1}$ is achievable for the system. This completes the
proof.
\end{proof}

\subsection{Three-user Gaussian Interference Channel:  $\text{DOF}=\frac{4}{3}$ is Achievable Almost Surely}
In this subsection, we consider the three-user GIC. First, the following model is defined as the standard model for
the channel.
\begin{definition}
The three user interference channel is called standard if it can be represented as
\begin{IEEEeqnarray}{rl}\label{alaki7}
y_1&=G_1x_1+x_2+x_3+z_1\nonumber\\
y_2&=G_2x_2+x_1+x_3+z_2\\
y_3&=G_3x_3+x_1+G_0x_2+z_3.\nonumber
\end{IEEEeqnarray}
where $x_i$ for User $i$ is subject to the power constraint $P$. $z_i$ at Receiver $i$ is AWGN with variance
$\sigma^2$.
\end{definition}
In the following lemma, it is proved that in fact characterizing the DOF of the standard channel causes no harm on
the generalization of the problem.
\begin{lemma}
For every three-user GIC there exists a standard channel with the same DOF.
\end{lemma}
\begin{proof}
The channel model is the special case of that of $K$-user GIC in (\ref{k-user model}) where $K=3$, i.e., the
input-output relation can be written as
\begin{IEEEeqnarray}{rl}
y_1&=h_{11}x_1+h_{12}x_2+h_{13}x_3+z_1\nonumber\\
y_2&=h_{21}x_1+h_{22}x_2+h_{23}x_3+z_2\\
y_3&=h_{31}x_1+h_{32}x_2+h_{33}x_3+z_3.\nonumber
\end{IEEEeqnarray}
Clearly, linear operations at transmitters and receivers do not affect the capacity region of the channel. Hence, we
adopt the following linear operations:
\begin{enumerate}
\item Transmitter 1 sends $x_1=\frac{h_{23}h_{12}}{h_{21}}\tilde{x}_1$ to the channel and Receiver 1 divides the
received signal by $h_{12}h_{13}$.
\item Transmitter 2 sends $x_2=h_{13}\tilde{x}_2$ to the channel and Receiver 2 divides the received signal by
$h_{12}h_{23}$.
\item Transmitter 3 sends $x_3=h_{12}\tilde{x}_3$ to the channel and Receiver 3 divides the received signal by
$\frac{h_{21}}{h_{12}h_{23}h_{31}}$.
\end{enumerate}
If $\tilde{y}_i$ for $i\in\{1,2,3\}$ denotes the output of Receiver $i$ after above operations then it is easy to see
that from input $\tilde{x}_i$ to output $\tilde{y}_i$ the channel behaves as (\ref{alaki7}), i.e., it can be written
as \begin{IEEEeqnarray}{rl}\label{alaki8}
\tilde{y}_1&=G_1\tilde{x}_1+\tilde{x}_2+\tilde{x}_3+\tilde{z}_1\nonumber\\
\tilde{y}_2&=G_2\tilde{x}_2+\tilde{x}_1+\tilde{x}_3+\tilde{z}_2\\
\tilde{y}_3&=G_3\tilde{x}_3+\tilde{x}_1+G_0\tilde{x}_2+\tilde{z}_3,\nonumber
\end{IEEEeqnarray}
where $\tilde{z}_i$ is the Gaussian noise at Receiver $i$ for $i\in\{1,2,3\}$ with variance $\sigma_i^2=\delta_i
\sigma^2$  where $\delta_i$ is constant depending on the channel coefficients. Similarly, the input power constraint
of Transmitter $i$ for $i\in\{1,2,3\}$ becomes $P_i=\gamma_i P$  where $\gamma_i$ is constant depending on the
channel coefficients. Moreover, the channel coefficients can be written as
\begin{IEEEeqnarray*}{rl}
&G_0=\frac{h_{13}h_{21}h_{32}}{h_{12}h_{23}h_{31}},\\
&G_1=\frac{h_{11}h_{12}h_{23}}{h_{12}h_{21}h_{13}},\\
&G_2=\frac{h_{22}h_{13}}{h_{12}h_{23}},\\
&G_3=\frac{h_{33}h_{12}h_{21}}{h_{12}h_{23}h_{31}}.
\end{IEEEeqnarray*}

Since the above operations change the input powers as well as the noise variances, the completion of the theorem
requires additional steps to make the power constraints as well as noise variances all equal. Notice that increasing
(resp. decreasing) the power and decreasing (resp. increasing) the noise variance enlarges (resp. shrinks) the
capacity region of the channel. Therefore, two channels are defined as follows. In the first channel with the same
input-output relation as of (\ref{alaki8}) the power constraints at all transmitters and the noise variances at all
receivers are set to $\max\{P_1,P_2,P_3\}$ and $\min\{\sigma_1^2,\sigma_2^2,\sigma_3^2\}$, respectively. Similarly,
in the second channel the power constraints and noise variances are set to $\max\{P_1,P_2,P_3\}$ and
$\min\{\sigma_1^2,\sigma_2^2,\sigma_3^2\}$, respectively. The capacity region of the channel is sandwiched between
that of these two channels. Moreover, at high power regimes the SNRs of these two channel differ by a constant
multiplicative factor. Hence, they share the same DOF and either of them can be used as the desired channel. This
completes the proof.
\end{proof}

Having the standard model, a special case that the total DOF of the channel can be achieved is identified in the
following theorem.
\begin{theorem}\label{theorem 3-user full}
If the channel gain $G_0$ in (\ref{alaki7}) is rational then the DOF of $\frac{3}{2}$ is achievable almost surely.
\end{theorem}
\begin{proof}
If $G_0$ is rational, then the set of cross gains at each receiver takes up one rational dimension. Applying Theorem
\ref{theorem-k-user-special-case} with $m=1$ gives the desired result.
\end{proof}

In general, the event of having rational $G_0$ has probability zero. The following theorem concerns the general case.
\begin{theorem}
The DOF of $\frac{4}{3}$ is achievable for the three-user GIC almost surely.
\end{theorem}
\begin{proof}
The encoding used to prove this theorem is asymmetrical. User 1 encodes two data streams while User 2 and 3  encode
only one data stream. In fact, the transmit constellation of Users 1,2, and 3 are $\mathcal{U}_1+G_0\mathcal{U}'_1$,
$\mathcal{U}_2$, and $\mathcal{U}_2$, respectively. It is assumed that $\mathcal{U}_1$, $\mathcal{U}'_1$,
$\mathcal{U}_2$, $\mathcal{U}_3$ are single layer constellation with points in $[-Q\ Q]$. We claim that each data
stream can carry data with multiplexing gain of $\frac{1}{3}$, and since there are four data streams, the DOF of
$\frac{4}{3}$ is achievable. To accommodate such rate $Q=\lfloor \gamma
P^{\frac{1-\epsilon}{2(3+\epsilon)}}\rfloor$ where $\gamma$ and $\epsilon$ are two arbitrary constants. The input
signals from Transmitters 1, 2, and 3 are  $x_1=A(u_1+G_0 u'_1)$, $x_2=Au_2$, and $x_3=Au_3$, respectively. $A$ is
the normalizing factor which controls the output power of all transmitters. It can be readily shown that there exists
a constant $\gamma'$ such that $A=\gamma' P^{\frac{2+2\epsilon}{2(3+\epsilon)}}$.

The decoding at Receivers are performed differently. The received signal at Receiver 1 can be represented as
\begin{equation}
y_1=A(G_1u_1+G_1G_0u'_1+I_1)+z_1,
\end{equation}
where $I_1=u_2+u_3$ is the interference caused by Users 2 and 3. Clearly $I_1\in[-2Q\ 2Q]$. Receiver 1 is interested
in both $u_1$ and $u'_1$ and performs the proposed decoding scheme for each of them separately. By applying Theorem
\ref{basic}, one can deduce that each of data streams $u_1$ and $u'_1$ can accommodate $\frac{1}{2}$ of multiplexing
gain.

The received signal at Receiver 2 can be represented as
\begin{equation}
y_2=A(G_2u_2+I_2+G_0u'_1)+z_2,
\end{equation}
where $I_2=u_1+u_3$ is the aligned part of the interference caused by Users 2 and 3 and $I_2\in[-2Q\ 2Q]$. Receiver 2
is interested in $u_2$ while $I_2$ and $u'_1$ are interference. An application of Theorem \ref{basic} shows that the
multiplexing gain of $\frac{1}{3}$ is achievable for data stream $u_2$.

Finally, the received signal at Receiver 3 can be represented as
\begin{equation}
y_3=A(G_3u_3+u_1+G_0I_3)+z_2,
\end{equation}
where $I_3=u'_1+u_2$ is the aligned part of the interference caused by Users 2 and 3 and $I_3\in[-2Q\ 2Q]$. Receiver
3 is interested in $u_3$ while $I_3$ and $u_1$ are interference. Again by using Theorem \ref{basic}, one can deduce
that the multiplexing gain of $\frac{1}{3}$ is achievable for data stream $u_3$. This completes the proof.
\end{proof}

\section{multi-layer Constellation}\label{sec multiple layer}
In this section, multi-layer constellations are incorporated in the encoding scheme. Here, the focus would be on the
symmetric three-user GIC. This channel is modeled by:
\begin{IEEEeqnarray}{rl}
y_1&=x_1+h(x_2+x_3)+z_1\nonumber\\
y_2&=x_2+h(x_3+x_1)+z_2\\
y_3&=x_3+h(x_1+x_2)+z_3\nonumber
\end{IEEEeqnarray}
where $x_i$ and $y_i$ are the transmit and the received signals of
User $i$, respectively. The additive noise $z_i$ for $i\in
\{1,2,3\}$ is Gaussian distributed with zero mean and variance
$\sigma^2$. Users are subject to the power constraints $P$.

This channel is among channels satisfying conditions of Theorem \ref{theorem 3-user full}. Hence, one can deduce that
 the total DOF of $\frac{3}{2}$ is achievable for this channel almost surely. The reason for considering the
symmetric case is to reveal some aspects of multi-layer constellations. In this section, we obtain an achievable DOF
for all channel gains. For example, it will be shown the multi-layer constellation is capable of achieving the total
DOF of $\frac{3}{2}$ for all irrational gains.

As pointed out in Section \ref{sec coding}, in multi-layer constellations, constellation
points are selected from points represented in the base $W\in\mathbb{N}$. Since the channel is symmetric, all
transmitters use the same constellation $\mathcal{U}$ in which a point can be represented as
\begin{equation}\label{multi-layer points}
u(\mathbf{b})=\sum_{k=0}^{L-1}{b_{l}W^{l}},
\end{equation}
where $b_{l}\in\{0,1,\ldots,a-1\}$ for all $l\in \{0,2,\ldots,L-1\}$. $\mathbf{b}$ represents the vector
$(b_{0},b_{1},\ldots,b_{L-1})$. $a$ is the factor which controls the number of constellation points. We assume $a<W$.
Therefore, all constellation points in
(\ref{multi-layer points}) are distinct and the size of the constellation is $|\mathcal{U}|=a^L$. Hence, the maximum
rate possible for this data stream is
bounded by $L\log a$.

A random codebook is generated by randomly choosing points form $\mathcal{C}$ using the uniform distribution. This
can be accomplished by imposing a uniform distribution on each $b_{l}$. The signal transmitted by User 1,2, and 3 are
respectively $x_1=Au(\mathbf{b})$, $x_2=Au(\mathbf{b}')$, and $x_3=Au(\mathbf{b}'')$. $A$ is the normalizing factor
and controls the output power.

\begin{remark}
The multi-layer constellation used in this paper has DC component. In fact, this component needs to be removed at all
transmitters. However, it only duplicates the achievable rate and has no effect as far as the DOF is concerned.
\end{remark}

To obtain $A$, one needs to compute the input power. Since $b_{l}$ and $b_{j}$ are independent for $l\neq j$, we have
the following chain of inequalities
\begin{IEEEeqnarray}{rl}
E[X_1^2]&=A^2 W^{2(L-1)} \sum_{l=0}^{L-1}E\left[b_{l}^2\right]
W^{-2l}\nonumber\\
&\leq A^2 W^{2(L-1)}\frac{(a-1)(2a-1)}{6}\sum_{l=0}^{\infty}W^{-2l}\nonumber\\
&\leq A^2 W^{2(L-1)}\frac{a^2}{3}\times\frac{1}{1-W^{-2}}\nonumber\\
&\leq \frac{A^2 a^2 W^{2L}}{W^2-1}.\nonumber
\end{IEEEeqnarray}
Hence, if $A=\frac{\sqrt{(W^{2}-1)P}}{aW^L}$ then
$E\left[{X_i^2}\right]\leq P$ which is the desired power constraint.

Due to the symmetry of the system, it suffices to analyze the first
user's performance. The received constellation signal at Receiver 1 can
be written as
\begin{equation}
y_1=A\sum_{l=0}^{L-1}\Big(b_{l}+hI_{l}\Big) W^{l}+z_1,
\end{equation}
where $I_{l}=b'_{l}+b''_{l}$ is the interference caused by Transmitters 1 and 2. Clearly, the interference is aligned
and $I_{l}\in\{0,1,\ldots,2(a-1)\}$. A point in the received constellation $\mathcal{U}_r$ can be represented as
\begin{equation}\label{noise free signal}
u_r(\mathbf{b},\mathbf{I})=A\sum_{l=0}^{L-1}\Big(b_{l}+hI_{l}\Big)
W^{l},
\end{equation}
where $\mathbf{I}$ represents the vector
$(I_{0},I_{1},\ldots,I_{L-1})$. As pointed out before the received constellation needs to satisfy Property $\Gamma$.
Here, Property $\Gamma$ translates into the following relation:
\begin{equation}\nonumber
\Gamma:~~u_r(\mathbf{b},\mathbf{I})\neq
u_r(\tilde{\mathbf{b}},\tilde{\mathbf{I}}) ~\text{iff}~
(\mathbf{b},\mathbf{I})\neq
(\tilde{\mathbf{b}},\tilde{\mathbf{I}}),
\end{equation}
which means that the receiver is able to extract both
$\mathbf{b}_1$ and $\mathbf{I}_1$ from the received constellation.

Using (\ref{lower bound on R}) to bound the achievable rate, the total DOF of the channel can be written as
\begin{IEEEeqnarray}{rl}
r_{\text{sum}} & = \lim_{P\rightarrow\infty}\frac{3R_1}{0.5\log P}\nonumber\\
&\geq \lim_{P\rightarrow\infty}\frac{3\left(\log |\mathcal{U}|-1-P_e\log |\mathcal{U}|\right)}{{0.5\log
P}}\nonumber\\
& = \lim_{P\rightarrow\infty} \frac{3L(1-P_e)\log a}{0.5\log P}\label{dof rational}
\end{IEEEeqnarray}
where $P_e$ depends on the minimum distance in the received constellation $d_{\min}$ as of (\ref{error probability}).
In fact, to obtain the maximum rate we need to select the design parameters $a$, $W$, and $L$. Selection of these
parameters needs to provide 1) Property $\Gamma$ in the received constellation, 2) exponential decrease in $P_e$ as
$P$ goes to infinity, 3) maximum achievable DOF of the system. In the following, we investigate the relation between
these factors for rational and irrational channel gains separately.

\subsection{Rational Channel Gains}
In this subsection, we prove the following theorem which provides an achievable DOF for the symmetric three-user GIC
with rational gains.

\begin{theorem}\label{thm rational}
The following DOF is achievable for the symmetric three-user GIC where the channel gain is rational, i.e.
$h=\frac{n}{m}$:
\begin{equation*}
r_{\text{sum}}=
\begin{cases}
\frac{3\log(n)}{\log(n(2n-1))} & \text{if $2n\geq m$,}\\
\frac{3\log(s+1)}{\log((s+1)(2s+1))} &\text{if $2n< m$ and $m=2s+1$,}\\
\frac{3\log(s)}{\log(2s^2-n)}&\text{if $2n< m$ and $m=2s$.}
\end{cases}
\end{equation*}

%

\end{theorem}


Since $h$ is rational, it can be represented as $h=\frac{n}{m}$
where $(m,n)=1$. In this case, Equation (\ref{noise free signal})
can be written as
\begin{equation}\label{noise free signal rational I}
u_r(\mathbf{b},\mathbf{I})=\frac{A}{m}\sum_{l=0}^{L-1}\Big(m
b_{l}+nI_{l}\Big) W^{l}.
\end{equation}

The theorem is proved by partitioning the set of rational numbers in three subsets and analyzing the performance of
the system in each of them. Let us first assume that Property $\Gamma$ holds for given $W$ and $a$.
To obtain the total DOF of the system, one needs to derive the
minimum distance in the received constellation. It is also easy to
show that $d_{\text{min}}=\frac{A}{m}$. Using (\ref{error
probability}), the bound on the error probability is
\begin{IEEEeqnarray}{rl}\nonumber
P_e & < \exp\left({-\frac{(W^{2}-1)P}{8(am\sigma)^2W^{2L}}}\right).
\end{IEEEeqnarray}
Let $L$ be set as
\begin{equation}\label{K assym}
L=\lfloor \frac{\log\left(P^{0.5-\epsilon}
\right)}{\log(W)}\rfloor,
\end{equation}
where $\epsilon>0$ is an arbitrary constant. Clearly, with this
choice of $K$, $P_2\leq \exp{\left(-\gamma P^{2\epsilon}\right)}$ where $\gamma$ is a constant. This results in
$P_e\rightarrow 0$ as $\text{SNR}\rightarrow \infty$. By using (\ref{dof rational}), the DOF of the system can be
derived as
\begin{IEEEeqnarray}{rl}\nonumber
r_{\text{sum}}&=\lim_{P\rightarrow \infty}
\frac{3L(1-P_e)\log a}{0.5\log P}\\\nonumber &
=\lim_{P\rightarrow \infty}
\frac{3L\log(a)}{0.5\log P}\\\nonumber
&=\lim_{P\rightarrow \infty} \frac{\lfloor
\frac{\log\left(P^{0.5-\epsilon}
\right)}{\log W}\rfloor\log a}{0.5\log P}\nonumber\\
&=\frac{\log a}{\log W}(1-2\epsilon).
\end{IEEEeqnarray}
Since $\epsilon$ can be chosen arbitrarily small, the DOF of the
system can be written as
\begin{equation}\label{dof last}
r_{\text{sum}}=\frac{3\log a}{\log W}.
\end{equation}

From (\ref{dof last}), one can deduce that in order to maximize the total DOF of the system one needs to maximize $a$
and minimize $W$ while respecting Property $\Gamma$. In fact, if it is possible to have $W=a^2$ then the upper bound
of $\frac{3}{2}$ can be touched. However, it is not possible in this case. The above theorem states that $W$ and $a$
can have the relation given in Table \ref{table 1}. Even though the relation is quadratic for all cases, the
achievable DOF is always below the upper bound.

\begin{table}[t]
\caption{Relation between $a$  and $W$ to satisfy Property $\Gamma$.}\label{table 1}
\centering
\small{\begin{tabular}{|c|c|c|c|}
  \hline
  &\textcolor{blue}{$h=n/m$} & \textcolor{blue}{$a$}  & \textcolor{blue}{$W$}   \\
  \hline
  \textcolor{red}{\textbf{Case I}} & $2n\geq m$ & $n$ & $n(2n-1)$  \\
  \hline
  \textcolor{red}{\textbf{Case II}} &$2n<m$ and $m=2s+1$ & $s+1$ & $(s+1)(2s+1)$ \\
  \hline
  \textcolor{red}{\textbf{Case III}} &$2n<m$ and $m=2s$ & $s$  & $2s^2-n$ \\
  \hline
\end{tabular}
}
\end{table}

To complete the proof of Theorem \ref{thm rational}, it is sufficient to prove that Property $\Gamma$ holds for the
cases given in Table \ref{table 1}.

\begin{lemma}
Property $\Gamma$ holds for all cases shown in Table \ref{table 1}.
\end{lemma}
\begin{proof}
This lemma is proved by induction on $L$. To show that the lemma holds for
$L=0$, it is sufficient to prove that the equation
\begin{equation}\label{diophantine eq1}
m(b_{0}-\tilde{b}_{0})+n(I_{0}-\tilde{I}_{0})=0
\end{equation}
has no nontrivial solution when
$b_{0},\tilde{b}_{0}\in\{0,1,\ldots,a-1\}$, and
$I_{0},\tilde{I}_{0}\in\{0,1,\ldots,2(a-1)\}$. In fact, two necessary conditions for the equation (\ref{diophantine
eq1}) to have a solution are $I_0-\tilde{I}_0$ is divisible by $m$ and $b_0-\tilde{b}_0$ is divisible by $n$. We can
prove that this equation has no solution if one of the two conditions does not hold. We consider each case
separately.

Case I: In this case $a=n$. Using the fact that
$-(n-1)\leq b_{0}-\tilde{b}_{0}\leq n-1$, one can deduce that $n\nmid (b_{0}-\tilde{b}_{0})$.

Case II: In this case $a=s+1$ where $m=2s+1$. Using the fact that
$-2s\leq I_{0}-\tilde{I}_{0}\leq 2s$, one can deduce that $m\nmid (I_{0}-\tilde{I}_{0})$.

Case III: In this case $a=s$ where $m=2s$. Using the fact that
$-2(s-1)\leq I_{0}-\tilde{I}_{0}\leq 2(s-1)$, one can deduce that $m\nmid (I_{0}-\tilde{I}_{0})$.

Now, it is assumed that the statement of the lemma holds for $L-1$.
To show it also holds for $L$, one needs to prove the equation
\begin{IEEEeqnarray}{rl}\label{diophantine eq2}
\frac{A}{m}\sum_{l=0}^{L}\Big(m(b_{l}-\tilde{b}_{l})+n(I_{l}-\tilde{I}_{l})\Big)
W^{l}=0
\end{IEEEeqnarray}
has no nontrivial solution. Equivalently, (\ref{diophantine eq2})
can be written as
\begin{IEEEeqnarray}{rl}
&m(b_{0}-\tilde{b}_{0})+n(I_{0}-\tilde{I}_{0})\nonumber\\
&=W \left(\sum_{l=0}^{L-1}\Big(m(b_{l+1}-\tilde{b}_{l+1})+n(I_{l+1}-\tilde{I}_{l+1})\Big)
W^{l}\right)\label{diophantine eq3}
\end{IEEEeqnarray}
In two steps, we prove that the above equation has no solution. First, it is assumed that the right hand side of
(\ref{diophantine eq3}) is zero. Due to inductive assumption, it
results in $b_{l}=\tilde{b}_{l}$ and $I_{l}=\tilde{I}_{l}$ for
all $l\in\{1,2,\ldots,L-1\}$. In addition, (\ref{diophantine eq3})
reduces to
\begin{IEEEeqnarray}{rl}
m(b_{0}-\tilde{b}_{0})+n(I_{0}-\tilde{I}_{0})=0
\end{IEEEeqnarray}
which is already shown that it has no solution except the trivial
one $b_{0}=\tilde{b}_{0}$ and $I_{0}=\tilde{I}_{0}$. Notice that this step holds for all three cases.

Second, it is assumed that the right hand side of (\ref{diophantine
eq3}) is non-zero. Now, (\ref{diophantine eq3}) can be written as
\begin{equation}\label{diophantine eq4}
m(b_{0}-\tilde{b}_{0})+n(I_{0}-\tilde{I}_{0})=cW,
\end{equation}
where $c\in\mathbb{Z}$ and $c\neq 0$. We prove that (\ref{diophantine eq4}) has no nontrivial solution in each three
cases.

Case I: Since $W=n(2n-1)$ in this case, $n$ divides $n(I_{0}-\tilde{I}_{0})$ as well as $cW$, but it can not
divide $m(b_{0}-\tilde{b}_{0})$ because $(m,n)=1$ and $-(n-1)\leq
b_{0}-\tilde{b}_{0}\leq n-1$. Hence, (\ref{diophantine eq4}) has a
solution if $b_{0}=\tilde{b}_{0}$ which contradicts the fact that
$n|I_{0}-\tilde{I}_{0}|< |c|W$.

Case II: In this case $W=(s+1)(2s+1)$ and $m=2s+1$. Hence, $2s+1$ divides both $m(b_{0}-\tilde{b}_{0})$ and $cW$
whereas it can not divide $n(I_{0}-\tilde{I}_{0})$. This is due to the fact that $(2n,m=2s+1)=1$ and $-2s\leq
I_{0}-\tilde{I}_{0}\leq 2s$. Hence, (\ref{diophantine eq4}) has a
solution if $I_{0}=\tilde{I}_{0}$ which contradicts the fact that
$m|b_{0}-\tilde{b}_{0}|< |c|W$.

Case III: In this case $W=2s^2-n$ and $m=2s$. Due to the symmetry and the
fact that
\begin{equation}
\left|m(b_{0}-\tilde{b}_{0})+n(I_{0}-\tilde{I}_{0})\right|<2W,
\end{equation}
it suffices to assume $l=1$. Substituting $W=2s^2-n$, Equation
(\ref{diophantine eq4}) can equivalently be written as
\begin{equation}
2s(b_{0}-\tilde{b}_{0})+n(I_{0}-\tilde{I}_{0}+1)=2s^2,
\end{equation}
It is easy to observe that $2s$ divides $2s(b_{0}-\tilde{b}_{0})$
as well as $2s^2$, but it can not divide
$n(I_{0}-\tilde{I}_{0}+1)$ because $(2s,n)=1$ and $-(2s-1)\leq
I_{0}-\tilde{I}_{0}\leq 2s-1$. Hence, (\ref{diophantine eq4}) has
a solution if $I_{0}+1=\tilde{I}_{0}$ which is impossible because
$2s|b_{0}-\tilde{b}_{0}|< 2s^2$. This completes the
proof.
\end{proof}

\subsection{Irrational Channel Gains}
In this subsection, it is shown that when the symmetric channel gain is
irrational then the total DOF of the system is achievable, i.e.,
$r_{\text{sum}}=\frac{3}{2}$. This result relies on a theorem in the field of Diophantine approximation due to 
Hurwitz. The theorem states as follows.

\begin{theorem}[Hurwitz \cite{hardy}]
There exist infinitely many solutions to the Diophantine equation
\begin{equation}\label{Hurwitz}
\mid \frac{n}{m}-h\mid < \frac{1}{m^2\sqrt{5}},
\end{equation}
where $h$ is an irrational number and $m,n\in\mathds{N}$.
\end{theorem}

Hurwitz's theorem approximates an irrational number by a rational
one and the goodness of the approximation is measured by the size of
the denominator.

\begin{theorem}
The total DOF of $\frac{3}{2}$ for the symmetric three-user GIC is achievable for all irrational channel gains.
\end{theorem}

\begin{remark}
This result can be readily extended to the symmetric $K$-user
GIC. In fact, it is easy to show that if
the symmetric channel gain is irrational, then $\frac{K}{2}$ is an
achievable DOF.
\end{remark}

For an irrational channel gain $h$, let us assume $m$ and $n$ are
two integers satisfying (\ref{Hurwitz}). Therefore,
$h=\frac{n}{m}+\delta$ where $|\delta|<\frac{1}{m^2\sqrt{5}}$. To
transmit data, $W$ is chosen as
\begin{equation}
W=\Big\lceil
\frac{2(1+2h)(a-1)}{\frac{1}{m}-4(a-1)|\delta|}\Big\rceil+1,
\end{equation}
where $a=\lfloor \frac{m^{1-\epsilon}\sqrt{5}}{4}\rfloor$ and
$\epsilon$ is an arbitrary positive number. The following chain of
inequalities shows that $W$ is positive.
\begin{IEEEeqnarray*}{rl}
4(a-1)|\delta|&\leq \frac{4(a-1)}{m^2\sqrt{5}}\\
&\leq\frac{4a}{m^2\sqrt{5}}\\
&\leq \frac{ m^{1-\epsilon}}{m^2}\\
&\leq \frac{1}{m}.
\end{IEEEeqnarray*}

In the following lemma, it is proved that the received
constellation possesses Property $\Gamma$.
\begin{lemma}
The received constellation in (\ref{noise free signal}) possesses
Property $\Gamma$.
\end{lemma}
\begin{proof}
Suppose there are $(\mathbf{b},\mathbf{I})$ and
$(\tilde{\mathbf{b}},\tilde{\mathbf{I}})$ such that their
corresponding constellation points are the same. Hence, we have
\begin{equation}
h=-\frac{m\sum_{l=0}^{L-1}(b_{l}-\tilde{b}_{l})W^{l}}{n\sum_{k=0}^{K}(I_{l}-\tilde{I}_{l})W^{l}},
\end{equation}
which is a contradiction, because the right hand side is a rational
number whereas the left hand side is an irrational number. This
completes the proof.
\end{proof}

To characterize the total DOF of the system, we need to derive the
minimum distance of points in the received constellation. In the
following lemma, the minimum distance is obtained.

\begin{lemma}
The minimum distance among the received constellation points with
$L$ levels of coding is lower-bounded as $d_{\text{min}}\geq
A\left(\frac{1}{m}-4(a-1)|\delta|\right)$.
\end{lemma}

\begin{proof}
This lemma is also proved by induction on $L$. In order to emphasize
that the minimum distance is a function of $L$, we may write
$d_{\min}(L)$. For $L=0$, we have
\begin{equation}
d_{\min}(0)=\min_{\Omega} ~A|\hat{b}_{0}-h\hat{I}_{0}|,
\end{equation}
where $\hat{b}_{0}=b_{0}-\tilde{b}_{0}$,
$\hat{I}_{0}=\tilde{I}_{0}-I_{0}$, and $\Omega$ is defined as
\begin{equation}\nonumber
\Omega=\{(\hat{b}_{0},\hat{I}_{0}): |\hat{b}_{0}|\leq
2(a-1),|\hat{I}_{0}|\leq 4(a-1)\}.
\end{equation}
Since $h=\frac{n}{m}+\delta$, we have
\begin{IEEEeqnarray}{rl}
d_{\min}(0)&=\min_{\Omega}
~A\left|\hat{b}_{0}-\frac{n}{m}\hat{I}_{0}-\delta
\hat{I}_{0}\right|\\
&\geq\min_{\Omega}
~A\left|\hat{b}_{0}+\frac{n}{m}\hat{I}_{0}\right|-\max_{\Omega}
~A|\delta \hat{I}_{0}|.
\end{IEEEeqnarray}
Since $|\hat{I}_{0}|\leq 4(a-1)$, we have
\begin{equation}
d_{\min}(0)\geq A\left(\frac{1}{m}-4(a-1)|\delta|\right),
\end{equation}
which is the desired result.

Now, it is assumed that the statement in the lemma holds for any $L-1$
level code. We need to show it also holds for $L$ level codes. The
difference between two distinct constellation points is written as
\begin{equation}\label{difference }
\Delta=AW\sum_{l=0}^{L-1}
(\hat{b}_{l+1}-h\hat{I}_{l+1})W^{l}+A(\hat{b}_{0}-h\hat{I}_{0}).
\end{equation}
Let us assume the first term in (\ref{difference }) is zero. In this
case, the minimum distance can be lower-bounded as
\begin{equation}
d_{\min}(L)\geq\min_{\Omega}
~A\left|\hat{b}_{0}-h\hat{I}_{0}\right|.
\end{equation}
The minimization problem is equivalent to that of case
$L=0$. Hence,
\begin{equation}
d_{\min}(L)\geq A\left(\frac{1}{m}-4(a-1)|\delta|\right),
\end{equation}
which is the desired result. If the first term in (\ref{difference
}) is non-zero, then its absolute value is at least $d_{\min}(L-1)$.
By the assumption of induction, we have
\begin{equation}
d_{\min}(L-1)\geq
A\left(\frac{1}{m}-4(a-1)|\delta|\right).
\end{equation}
Therefore, we can obtain the following chain of inequalities
\begin{IEEEeqnarray}{rl}
d_{\min}(K)&=\min_{} |\Delta|\nonumber\\
&\geq Wd_{\min}(K-1)-\max_{}A\left|\hat{b}_{0}-h\hat{I}_{0}\right|\nonumber\\
&\geq AW(\frac{1}{m}-4(a-1)|\delta|)\nonumber\\
&~~~~~~~~~~~~~~~~-2A(1+2h)(a-1)\nonumber\\
&\geq A (\frac{1}{m}-4(a-1)|\delta|)\times\nonumber\\
&~~~~~~~~~~~~~~~~\left(W-\frac{2(1+2h)(a-1)}{\frac{1}{m}-4(a-1)|\delta|}\right)\nonumber\\
&\geq A (\frac{1}{m}-4(a-1)|\delta|).\nonumber
\end{IEEEeqnarray}
This completes the proof.
\end{proof}

Having a lower bound on the minimum distance, we can derive an upper
bound for the error probability as follows
\begin{IEEEeqnarray}{rl}\nonumber
P_e & < \exp{\left(\frac{d_{\min}^2}{8\sigma^2}\right)}\nonumber\\
&\leq \exp\left({-\frac{A^2(\frac{1}{m}-4(a-1)|\delta|)^2}{8\sigma^2}}\right).
\end{IEEEeqnarray}

Due to Hurwitz's theorem, there are infinitely many solutions for
(\ref{Hurwitz}), i.e., there is a sequence of $m$ converging to
infinity and satisfying (\ref{Hurwitz}). Therefore, there exists a
sequence of $P$'s converging to infinity and satisfying
$m=\lfloor\log(P)\rfloor$. We take the limit in (\ref{lower bound on R}) with respect to this sequence. $L$ is again
chosen as
\begin{equation}
L=\lfloor \frac{\log\left(P^{0.5-\epsilon}
\right)}{\log(W)}\rfloor,
\end{equation}

To show that $P_e$ decays exponentially with respect to $P$, we consider the following chain of inequalities
\begin{IEEEeqnarray}{rl}
P_e &\leq \exp\left(-\frac{(W^{2}-1)P}{8a^2\sigma^2W^{2L}}(\frac{1}{m}-4(a-1)|\delta|)^2\right)\nonumber\\
 &\leq \exp\left(-\frac{W^{2}-1}{8a^2\sigma^2}(\frac{1}{m}-4(a-1)|\delta|)^2P^{2\epsilon}\right)\nonumber\\
&\stackrel{(a)}{\simeq}
\exp\left(- \gamma P^{2\epsilon} \right)\rightarrow
0~ \text{as}~ P\rightarrow \infty\nonumber
\end{IEEEeqnarray}
where (a) comes from the fact that $\frac{W^{2}-1}{8a^2\sigma^2}(\frac{1}{m}-4(a-1)|\delta|)^2$ approaches a
constant, say $\gamma$, as $P\rightarrow \infty$. The total DOF can be calculated using (\ref{lower bound on R}) as
follows
\begin{IEEEeqnarray}{rl}\nonumber
r_{\text{sum}}&=\lim_{P\rightarrow \infty}
 \frac{3L\log(a)}{0.5\log P}\nonumber\\
&=\lim_{P\rightarrow \infty} \frac{3\log(a)}{\log(W)}(1-2\epsilon)\nonumber\\
&=\frac{3}{2}(1-\epsilon)(1-2\epsilon).\nonumber
\end{IEEEeqnarray}
Since $\epsilon$ can be chosen arbitrarily small, $r_{\text{sum}}=\frac{3}{2}$
is achievable.

\section{Conclusion}
We proposed a novel coding scheme in which data is modulated using constellation carved from rational points and
directed by multiplying by irrational numbers. Using tools from the filed of Diophantine approximation in number
theory, in particular the Khintchine-Groshev and Hurwitz theorems, we proved that the proposed coding scheme achieves
the total DOF of several channels. We considered the single layer and multi-layer constellations for the encoding
part.

Using the single layer constellation, we proved that the time-invariant two-user $X$ channel and three-user GIC
achieve the DOF of $\frac{4}{3}$ alike. However, for the former it meets the upper bound which means that the total
DOF of the two-user $X$ channel is established. This is the first example in which it is shown that a time invariant
single antenna system does not fall short of achieving its total DOF.

Using the multi-layer constellation, we derived an achievable DOF for the symmetric three-user GIC. We showed that
this achievable DOF is an everywhere discontinuous function with respect to the channel gain. In particular, we
proved that for the irrational channel gains the achievable DOF meets the upper bound $\frac{3}{2}$ and for the
rational gains, even by allowing carry over from multiple layers, the achievable DOF has a gap to the available upper
bounds.


\end{document}